\documentclass[10pt]{extarticle}

\usepackage{times}
\usepackage{amssymb,amsmath,amsthm}
\usepackage{mathrsfs}
\usepackage{epsfig}
\usepackage{color}
\usepackage[latin1]{inputenc}

\newcommand\ie{{\em i.e.}}

\def\B{\mathscr B}

\def\d{\mathrm{d}}
\def\F{\mathscr F}
\def\G{\mathcal G}

\def\H{\mathcal H}

\def\K{\mathscr K}

\def\N{\mathbb N}
\def\R{\mathbb R}
\def\S{\mathbb S}
\def\SS{\mathscr S}

\def\Hrond{\mathscr H}
\def\HS{\mathfrak h}
\def\Pv{\mathrm{Pv}}

\def\lone{\mathsf{L}^{\:\!\!1}}

\def\ltwo{\mathsf{L}^{\:\!\!2}}
\def\linf{\mathsf{L}^{\:\!\!\infty}}
\def\e{\mathop{\mathrm{e}}\nolimits}

\def\slim{\mathop{\hbox{\rm s-}\lim}\nolimits}

\newtheorem{Theorem}{Theorem}[section]

\newtheorem{Lemma}[Theorem]{Lemma}

\begin{document}

%--------------------------------------------------------------------------------------
% Title
%--------------------------------------------------------------------------------------

\title{Explicit formulas for the Schr\"odinger wave operators in $\R^2$}

\author{S. Richard$^1$~~and R. Tiedra de
Aldecoa$^2$\footnote{Supported by the Chilean Fondecyt Grant 1090008
and by the Iniciativa Cientifica Milenio ICM P07-027-F ``Mathematical Theory of
Quantum and Classical Magnetic Systems'' from the Chilean Ministry of Economy.}}

\date{\small}
\maketitle \vspace{-1cm}

\begin{quote}
\emph{
\begin{itemize}
\item[$^1$] Universit\'e de Lyon; Universit\'e
Lyon 1; CNRS, UMR5208, Institut Camille Jordan, \\
43 blvd du 11 novembre 1918, F-69622
Villeurbanne-Cedex, France.
\item[$^2$] Facultad de Matem\'aticas, Pontificia Universidad Cat\'olica de Chile,\\
Av. Vicu\~na Mackenna 4860, Santiago, Chile
\item[] \emph{E-mails:} richard@math.univ-lyon1.fr, rtiedra@mat.puc.cl
\end{itemize}
}
\end{quote}

%--------------------------------------------------------------------------------------
% Abstract
%--------------------------------------------------------------------------------------

\begin{abstract}
In this note, we derive explicit formulas for the Schr\"odinger wave operators in
$\R^2$ under the assumption that $0$-energy is neither an eigenvalue nor a
resonance. These formulas justify the use of a recently introduced topological
approach of scattering theory to obtain index theorems.
\end{abstract}


\smallskip


%--------------------------------------------------------------------------------------
\section{Introduction and main theorem}
\setcounter{equation}{0}
%--------------------------------------------------------------------------------------

It has recently been shown that introducing $C^*$-algebraic methods in scattering
theory leads naturally to some index theorems. The starting point for this approach was
the observation made in \cite{KR06} (see also \cite{BS12,IR12,KR08,KR12,RT10}) that
Levinson's theorem can be reinterpreted as an index theorem.
In its original form, Levinson's theorem establishes roughly the equality between the number of bound states of a Schr\"odinger operator and an expression involving the scattering operator for the underlying physical system.
The main idea of the new approach consists in showing that the corresponding wave operators $W_\pm$ belong to a suitable $C^*$-algebra, and then in applying technics of non-commutative topology and cyclic cohomology to obtain an index theorem.
For more complex scattering systems, other topological equalities involving index theorems for families
as well as higher degree traces can also be derived (see \cite{KPR}).

To apply the new approach, a rather good understanding of the operators $W_\pm$ is necessary. Indeed, these partial isometries (which are also Fredholm operators under rather weak assumptions) have to be
affiliated to the central $C^*$-algebra of a short exact sequence, with the algebra of compact operators as an ideal and an understandable quotient algebra. For that purpose, explicit formulas for $W_\pm$ have been exhibited for various models of quantum mechanics (see for example \cite{PR11,RT13}). In the present note, we add to this list an explicit formula for the Schr\"odinger wave operators in $\R^2$ in the generic case. We recall that the $2$-dimensional case presents various difficulties and deserves a special attention; see the seminal works \cite{BGD88,JY02,Yaj99} and references therein, or the more recent papers \cite{EG12_0,EG12_1,KMRT05,KMRT09,Sch05,W11}.

So, let us be more precise about our result. We consider in the Hilbert space $\H:=\ltwo(\R^2)$ the free Schr\"odinger operator $H_0:=-\Delta$ and the perturbed operator $H:=-\Delta+V$, with a potential $V\in\linf(\R^2;\R)$ decaying fast enough at infinity. In such a situation,
the wave operators
\begin{equation}\label{wave}
W_\pm:=\slim_{t\to\pm\infty}\e^{itH}\e^{-itH_0}
\end{equation}
exist and are asymptotically complete. As a consequence, the scattering operator $S:=W_+^* W_-$ is a unitary operator in $\H$.
If $\B(\H)$ (resp. $\K(\H)$) denotes the set of bounded (resp.
compact) operators in $\H$, and if $A$ stands for the generator of dilations in $\R^2$, then our
main result is the following\;\!:

\begin{Theorem}\label{Java}
Suppose that $V$ satisfies $|V(x)|\le{\rm Const.}\;\!(1+|x|)^{-\sigma}$ with
$\sigma>11$ for almost every $x\in\R^2$, and assume that $H$ has neither eigenvalues
nor resonances at $0$-energy. Then, one has in $\B(\H)$ the equalities
\begin{equation}\label{jolieformule}
W_-=1+R(A)(S-1)+K
\qquad\hbox{and}\qquad
W_+=1+\big(1-R(A)\big)(S^*-1)+K',
\end{equation}
with $R(A):=\frac12\big(1+\tanh(\pi A/2)\big)$ and $K,K'\in\K(\H)$.
\end{Theorem}

We stress that the absence of eigenvalues or resonances at $0$-energy is generic.
Their presence leads to slightly more complicated expressions and will be considered
elsewhere. On the other hand, we note that no spherical symmetry is imposed on $V$.
The rest of the text is devoted to the proof of formulas \eqref{jolieformule} as well
as another formula for $W_\pm$ which does not involve any compact remainder (see
Theorem \ref{BigMama}).\\

\noindent
{\bf Notations\hspace{1pt}:}
$\N:=\{0,1,2,\ldots\}$ is the set of natural numbers, $\R_+:=(0,\infty)$,
$\langle x\rangle :=\sqrt{1+|x|^2}$ and $\SS$ is the Schwartz space on $\R^2$. The
sets $\H^s_t$ are the weighted Sobolev spaces over $\R^2$ with index $s\in\R$
associated with derivatives and index $t\in\R$ associated with decay at infinity
\cite[Sec.~4.1]{ABG} (with the convention that $\H^s:=\H^s_0$ and $\H_t:=\H^0_t$).
For any $s,t\in\R$, the $2$-dimensional Fourier transform $\F$ is a topological
isomorphism of $\H^s_t$ onto $\H^t_s$, and the scalar product
$\langle\;\!\cdot\;\!,\;\!\cdot\;\!\rangle_\H$ extends continuously to a duality
$\langle\;\!\cdot\;\!,\;\!\cdot\;\!\rangle_{\H^s_t,\H^{-s}_{-t}}$ between $\H^s_t$
and $\H^{-s}_{-t}$. Given two Banach spaces $\G_1$ and $\G_2$, $\B(\G_1,\G_2)$ (resp.
$\K(\G_1,\G_2)$) stands for the set of bounded (resp. compact) operators from $\G_1$
to $\G_2$. Finally, $\otimes$ (resp. $\odot$) stands for the closed (resp. algebraic)
tensor product of Hilbert spaces or of operators.

%--------------------------------------------------------------------------------------
\section{Explicit formulas for the wave operators}\label{Sec_one}
\setcounter{equation}{0}
%--------------------------------------------------------------------------------------

Throughout this note, we use the Hilbert spaces $\H:=\ltwo(\R^2)$, $\HS:=\ltwo(\S)$, $\Hrond:=\ltwo\big(\R_+;\HS\big)$ and the unitary operator (spectral transformation)
$\F_0:\H\to\Hrond$ satisfying
$
(\F_0H_0f)(\lambda)
=\lambda(\F_0 f)(\lambda)
\equiv(L\F_0 f)(\lambda)
$
for $f\in\H^2$, a.e. $\lambda\in\R_+$, and $L$ the maximal multiplication operator in
$\Hrond$ by the variable in $\R_+$. The explicit formula for $\F_0$ is
\begin{equation}\label{def_F_0}
\big((\F_0 f)(\lambda)\big)(\omega)
=2^{-1/2}(\F f)(\sqrt\lambda\;\!\omega),
\quad f\in\SS,~\lambda\in\R_+,~\omega\in\S.
\end{equation}

In stationary scattering theory one defines the wave operators $W_\pm$ in terms of
suitable limits of the resolvents of $H_0$ and $H$ near the real axis. We shall
mainly use this approach, noting that for potentials $V$ as in Theorem \ref{Java}
both definitions for the operators $W_\pm$ coincide (see \cite[Thm.~5.3.6]{Yaf92}).
So, starting from \cite[Eq.~2.7.5]{Yaf92} and taking into account the resolvent
formula written in the symmetrized form \cite[Eq.~4.3]{JN01}, one obtains for
suitable $\varphi,\psi\in\Hrond$ (precise conditions are given in Theorem
\ref{BigMama} below) that
\begin{align}
&\big\langle\F_0(W_\pm-1)\F_0^*\varphi,\psi\big\rangle_{\!\Hrond}\nonumber\\
&=-\int_\R\d\lambda\,\lim_{\varepsilon\searrow0}\int_0^\infty\d\mu\,
\big\langle\big\{\F_0vM_0(\lambda\mp i\varepsilon)^{-1}v\;\!\F_0^*
\delta_\varepsilon(L-\lambda)\varphi\big\}(\mu),(\mu-\lambda\mp i\varepsilon)^{-1}
\psi(\mu)\big\rangle_\HS.\label{start}
\end{align}
with
$
\delta_\varepsilon(L-\lambda)
:=\frac\varepsilon\pi (L-\lambda +i\varepsilon)^{-1}(L-\lambda-i\varepsilon)^{-1}
$,
$v:=|V|^{1/2}$, $M_0(z):=u+vR_0(z)v$ and $u(x):=1$ if $V(x)\ge0$ while $u(x)=-1$ if
$V(x)<0$.

In order to exchange the integral over $\mu$ and the limit $\varepsilon\searrow0$, we
need a series of preparatory lemmas. First, we recall that the operator
$\F_0(\lambda):\SS\to\HS$ given by $\F_0(\lambda)f:=(\F_0f)(\lambda)$ extends to an
element of $\B(\H^s_t,\HS)$ for each $s\in\R$ and $t>1/2$, and that the map
$\R_+\ni\lambda\mapsto\F_0(\lambda)\in\B(\H^s_t,\HS)$ is continuous. We also have the
following result which is a direct consequence of what precedes and the estimate
\cite[Thm.~1.1.4]{Yaf10}\;\!:

\begin{Lemma}\label{lem1}
Let $s\ge0$ and $t>1/2$. Then, the function
$
\R_+\ni\lambda\mapsto\langle\lambda\rangle^{1/4}\F_0(\lambda)\in\B(\H^s_t,\HS)
$
is continuous and bounded.
\end{Lemma}

One also obtains the following result (whose
proof is analogous to the one of \cite[Lemma~2.2]{RT13})\;\!:

\begin{Lemma}\label{lem2}
Let $s >-1/2$ and $t>1$. Then, $\F_0(\lambda)\in\K(\H^s_t,\HS)$ for each
$\lambda\in\R_+$, and the function
$\R_+\ni\lambda\mapsto\F_0(\lambda)\in\K(\H^s_t,\HS)$ is continuous, admits a limit
as $\lambda\searrow0$ and vanishes as $\lambda\to\infty$.
\end{Lemma}

From now on, we use the notation $C_{\rm c}(\R_+;\G)$ for the set of compactly
supported continuous functions from $\R_+$ to some Hilbert space $\G$. With this
notation and what precedes, we note that the multiplication operator
$N:C_{\rm c}(\R_+;\H^s_t)\to\Hrond$ given by
\begin{equation}\label{defdeN}
(N\xi)(\lambda):=\F_0(\lambda)\;\!\xi(\lambda),
\quad\xi\in C_{\rm c}(\R_+;\H^s_t),~\lambda\in\R_+,
\end{equation}
extends for $s>-1/2$ and $t>1$ to an element of
$\B\big(\ltwo(\R_+;\H^s_t),\Hrond\big)$. We also note that the limit
$\varepsilon\searrow0$ of the operator $\F_0^*\delta_\varepsilon(L-\lambda)$
appearing in \eqref{start} satisfies the following (see \cite[Lemma~2.3]{RT13} for a
proof)\;\!:

\begin{Lemma}\label{lemlimite}
For $s\ge0$, $t>1$, $\lambda\in\R_+$ and $\varphi\in C_{\rm c}(\R_+;\HS)$, one has
$
\lim_{\varepsilon\searrow0}\F_0^*\;\!\delta_\varepsilon(L-\lambda)\varphi
=\F_0(\lambda)^*\varphi(\lambda)
$
in $\H^{-s}_{-t}$.
\end{Lemma}

The next necessary result concerns the limit
$
M_0(\lambda+i0)^{-1}:=\lim_{\varepsilon\searrow0}M_0(\lambda + i\varepsilon)^{-1}
$,
$\lambda\in\R_+$ (a similar result holds for $M_0(\lambda-i0)^{-1}$). First, we
recall that $H$ does not have positive eigenvalues \cite[Sec.~1]{Kat59}. Therefore,
for $V$ as in Theorem \ref{Java}, one infers from the limiting absorption principles
for $H_0$ and $H$ \cite[Thm.~4.2]{Agm75} the existence in $\B(\H)$ of the limits
$M_0(\lambda+i0):=\lim_{\varepsilon\searrow0}\big(u+vR_0(\lambda+i\varepsilon)v\big)$
and $M(\lambda+i0):=\lim_{\varepsilon\searrow0}\big(u-vR(\lambda+i\varepsilon)v\big)$,
and their continuity with respect to $\lambda$. This, together with the fact that
$uM(\lambda+i\varepsilon)u=M_0(\lambda+i\varepsilon)^{-1}$ for $\varepsilon>0$,
implies the existence and the continuity of the map
$\R_+\ni\lambda\mapsto M_0(\lambda+i0)^{-1}\in\B(\H)$. Also, one has
$\lim_{\lambda\to\infty}M_0(\lambda+i0)^{-1}=u$ in $\B(\H)$, since
$\lim_{\lambda\to\infty}vR_0(\lambda+i0)v=0$ in $\B(\H)$ \cite[Prop.~7.1.2]{Yaf10}.
On the other hand, the existence in $\B(\H)$ of the limit
$\lim_{\lambda\searrow0}M_0(\lambda +i0)^{-1}$ (which has been studied in detail in
\cite{JN01}) highly depends on the presence or absence of eigenvalues or resonances
at $0$-energy; the limit does not exist in their presence, but in the generic case
(\ie~in the absence of eigenvalues or resonances at $0$-energy) the limit exists
\cite[Eq.~(6.55)]{JN01}. With this information, we obtain the following\;\!:

\begin{Lemma}\label{lem_on_sigma}
Let $V$, $\sigma$ and $H$ be as in Theorem \ref{Java}. Then, the map
$
\R_+\ni\lambda\mapsto M_0(\lambda+i0)^{-1}\in\B(\H)
$
is continuous and bounded. Furthermore, the multiplication operator
$
B:C_{\rm c}\big(\R_+;\HS\big)\to\ltwo\big(\R_+;\H_{\sigma/2}\big)
$
given by
\begin{equation}\label{defdeB}
(B\varphi)(\lambda)
:=vM_0(\lambda+i0)^{-1}v\;\!\F_0(\lambda)^*\varphi(\lambda)\in\H_{\sigma/2},
\quad\varphi\in C_{\rm c}\big(\R_+;\HS\big),~\lambda\in\R_+,
\end{equation}
extends to an element of $\B\big(\Hrond,\ltwo(\R_+;\H_{\sigma/2})\big)$.
\end{Lemma}

\begin{proof}
The condition $\sigma>11$ is imposed in order to fulfill the assumptions of
\cite[Thm.~6.2 \& Eq.~(6.55)]{JN01} for the existence of the limit
$\lim_{\lambda\searrow0}M_0(\lambda+i0)^{-1}$ in $\B(\H)$. The continuity and the
boundedness of the map $\R_+\ni\lambda\mapsto M_0(\lambda+i0)^{-1}\in\B(\H)$ follow
then from what has been said before. Finally, the second part of the statement is a
consequence of what precedes and Lemma \ref{lem1}.
\end{proof}

Before deriving our first formula for $W_-$, we recall that the dilation group
$\{U^+_\tau\}_{\tau\in\R}$ in $\ltwo(\R_+)$ with self-adjoint generator $A_+$ is
given by $\big(U^+_\tau\varphi\big)(\lambda):=\e^{\tau/2}\varphi(\e^\tau\lambda)$ for
$\varphi\in C_{\rm c}(\R_+)$, $\lambda\in\R_+$ and $\tau\in\R$. We also introduce the
function $\vartheta\in C(\R)\cap\linf(\R)$ given by
$\vartheta(\nu):=\frac12\big(1-\tanh(\pi\nu)\big)$ for $\nu\in\R$. Finally, we recall
that the Hilbert spaces $\ltwo(\R_+;\H^s_t)$ and $\Hrond$ can be naturally identified
with the Hilbert spaces $\ltwo(\R_+)\otimes\H^s_t$ and $\ltwo(\R_+)\otimes\HS$.

\begin{Theorem}\label{BigMama}
Let $V$, $\sigma$ and $H$ be as in Theorem \ref{Java}. Then, one has in $\B(\Hrond)$
the equality
\begin{equation}\label{eqequal}
\F_0(W_--1)\;\!\F_0^*
=-2\pi iN\;\!\big\{\vartheta(A_+)\otimes1_{\H_{\sigma/2}}\big\}B,
\end{equation}
with $N$ and $B$ defined in \eqref{defdeN} and \eqref{defdeB}.
\end{Theorem}

The proof below consists in two parts. First, we show that the expression
\eqref{start} is well-defined for $\varphi$ and $\psi$ in dense subsets of $\Hrond$,
and then we prove the stated equality.

\begin{proof}
Take $\varphi\in C_{\rm c}(\R_+;\HS)$ and $\psi\in C_{\rm c}^\infty(\R_+)\odot C(\S)$,
and set $t:=\sigma/2 $. Then, we have for each $\varepsilon>0$ and $\lambda\in\R_+$
the inclusions
$
g_\varepsilon(\lambda)
:=vM_0(\lambda+i\varepsilon)^{-1}v\;\!\F_0^*\delta_\varepsilon(L-\lambda)\varphi\in\H_t
$
and $f(\lambda):=\F_0(\lambda)^*\psi(\lambda)\in\H_{-t}$. So, using the formula
$
(\mu-\lambda+i\varepsilon)^{-1}
=-i\int_0^\infty\d z\e^{i(\mu-\lambda)z}\e^{-\varepsilon z}
$
and then applying Fubini's theorem, one obtains that \eqref{start} is equal to
\begin{align}
-i\;\!\lim_{\varepsilon\searrow0}\int_0^\infty\d z\,\e^{-\varepsilon z}
\bigg\langle g_\varepsilon(\lambda),\int_{-\lambda}^\infty\d\nu\,\e^{i\nu z}
f(\nu+\lambda)\bigg\rangle_{\H_t,\H_{-t}}.\label{eq2}
\end{align}
Now, we know from Lemma \ref{lemlimite} and the paragraph following it that
$g_\varepsilon(\lambda)$ converges to
$
g_0(\lambda):=v M_0(\lambda+i0)^{-1}v\F_0^*(\lambda)\varphi(\lambda)
$
in $\H_t$ as $\varepsilon\searrow0$. Therefore, the family
$\|g_\varepsilon(\lambda)\|_{\H_t}$ can be bounded independently of $\varepsilon\in(0,1)$,
and thus the absolute value of the integrant in \eqref{eq2} can also be bounded independently of $\varepsilon\in(0,1)$.

To exchange the limit $\lim_{\varepsilon\searrow0}$ and the integral over $z$ in
\eqref{eq2}, it remains to show that $z\mapsto\big\|\int_{-\lambda}^\infty\d\nu\,
\e^{i\nu z}f(\nu+\lambda)\big\|_{\H_{-t}}\in \lone(\R_+,\d z)$. For that purpose, we write
$h_\lambda$ for the trivial extension
of the function
$
(-\lambda,\infty)\ni\nu\mapsto f(\nu+\lambda)\in\H_{-t}
$
to all of $\R$, and then note that $\big\|\int_{-\lambda}^\infty\d\nu\,
\e^{i\nu z}f(\nu+\lambda)\big\|_{\H_{-t}}$ can be rewritten
as $(2\pi)^{1/2}\|(\F_1^*h_\lambda)(z)\|_{\H_{-t}}$, with $\F_1$ the $1$-dimensional
Fourier transform. Furthermore, if $P_1$ denotes the
self-adjoint operator $-i\frac\d{\d z}$ on $\R$, then
$$
\big\|\big(\F_1^*h_\lambda\big)(z)\big\|_{\H_{-t}}
=\langle z\rangle^{-2}
\big\|\big(\F_1^*\langle P_1\rangle^2h_\lambda\big)(z)\big\|_{\H_{-t}},
\quad z\in\R_+\;\!.
$$
So, one would have that
$
z \mapsto \|(\F_1^*h_\lambda)(z)\|_{\H_{-t}}\in\lone(\R_+,\d z)
$
if
$
\big\|\big(\F_1^*\langle P_1\rangle^2h_\lambda\big)(z)\big\|_{\H_{-t}}
$
were bounded independently of $z$. Now, if $\psi=\eta\otimes\xi$ with
$\eta\in C_{\rm c}^\infty(\R_+)$ and $\xi\in C(\S)$, then one has for any $x\in\R^2$
$$
\big(f(\nu+\lambda)\big)(x)
=\frac1{2^{3/2}\pi}\;\!\eta(\nu+\lambda)\int_\S\d\omega\,
\e^{i\sqrt{\nu+\lambda}\;\!\omega\cdot x}\xi(\omega),
$$
which in turns implies that
$\big|\big\{\big(\F_1^*\langle P_1\rangle^2h_\lambda\big)(z)\big\}(x)\big|
\le{\rm Const.}\;\!\langle x\rangle^2$,
with a constant independent of $x\in\R^2$ and $z\in\R_+$. Since the r.h.s. belongs to
$\H_{-t}$ for $t>3$, one deduces that
$\big\|\big(\F_1^*\langle P_1\rangle^2h_\lambda\big)(z)\big\|_{\H_{-t}}$ is bounded
independently of $z$ for each $\psi=\eta\otimes\xi$, and thus for each
$\psi\in C^\infty_{\rm c}(\R_+)\odot C(\S)$ by linearity. As a consequence, one can
    apply Lebesgue dominated convergence theorem in \eqref{eq2}, and thus conclude that \eqref{start} is equal to
$\big\langle\F_0(W_\pm-1)\;\!\F_0^*\varphi,\psi\big\rangle_{\!\Hrond}$ on the sets of vectors $\varphi\in C_{\rm c}(\R_+;\HS)$
and $\psi\in C_{\rm c}^\infty(\R_+)\odot C(\S)$.

The next task is to prove \eqref{eqequal}. Let $\chi_+$ denote the characteristic function for $\R_+$. Then, a computation as in the proof of \cite[Thm.~2.6]{RT13} shows, in the sense of
distributions with values in $\H_{-t}$, that
\begin{align*}
\int_0^\infty\d z\int_\R\d\nu\;\!\e^{i\nu z}h_\lambda(\nu)
&=\sqrt{2\pi}\int_\R\d\mu\,\big(\F_1^*\chi_+\big)\big(\lambda(\e^\mu-1)\big)
\;\!\lambda\e^{\mu/2}\big\{\big(U_\mu^+\otimes1_{\H_{-t}}\big)f\big\}(\lambda)\\
&=\int_\R\d\mu\left(\pi\;\!\delta_0(\e^\mu-1)
+i\,\Pv\frac{\e^{\mu/2}}{\e^\mu-1}\right)
\big\{\big(U_\mu^+\otimes1_{\H_{-t}}\big)f\big\}(\lambda),
\end{align*}
with $\delta_0$ the Dirac delta distribution and $\Pv$ the principal value. So,
using successively the identity $\frac{\e^{\mu/2}}{\e^\mu-1}=\frac1{2\sinh(\mu/2)}$,
the equality \cite[Table 20.1]{Jef95}
$
\big(\F_1\vartheta\big)(\mu)
=\sqrt{\frac\pi2}\;\!\delta_0\big(\e^\mu-1\big)
+\frac i{2\sqrt{2\pi}}\;\!\Pv\frac1{\sinh(\mu/2)}
$
and the equation
$
\big\{\vartheta(A_+)\otimes1_{\H_{-t}}\big\}f
=\frac1{\sqrt{2\pi}}\int_\R\d\mu\,\big(\F_1\vartheta\big)(\mu)
\big(U_\mu^+\otimes 1_{\H_{-t}}\big)f
$,
one infers that
\begin{align*}
\big\langle\F_0(W_--1)\;\!\F_0^*\varphi,\psi\big\rangle_{\!\Hrond}
&=i\;\!\sqrt{2\pi}\int_{\R_+}\d\lambda\left\langle g_0(\lambda),\int_\R\d\mu\,
\big(\F_1\vartheta\big)(\mu)\;\!
\big\{\big(U_\mu^+\otimes1_{\H_{-t}}\big)f\big\}(\lambda)
\right\rangle_{\H_t,\H_{-t}} \\
&=\big\langle-2\pi iN\;\!\big\{\vartheta(A_+)\otimes1_{\H_{t}}\big\}B\varphi,
\psi\big\rangle_{\!\Hrond}.
\end{align*}
This concludes the proof, since the sets of vectors $\varphi\in C_{\rm c}(\R_+;\HS)$
and $\psi\in C_{\rm c}^\infty(\R_+)\odot C(\S)$ are dense in $\Hrond$.
\end{proof}

We now recall a final lemma which is essential for Theorem \ref{Java}. Its proof
is identical to the proof of \cite[Lemma~2.7]{RT13}.

\begin{Lemma}\label{Lemma_compact}
Take $s>-1/2$ and $t>1$. Then, the difference
$
\big\{\vartheta(A_+)\otimes1_{\HS}\big\}N
-N\big\{\vartheta(A_+)\otimes 1_{\H^s_t}\big\}
$
belongs to $\K\big(\ltwo(\R_+;\H_t^s),\Hrond\big)$.
\end{Lemma}

\begin{proof}[Proof of Theorem \ref{Java}]
Theorem \ref{BigMama}, Lemma \ref{Lemma_compact} and the identity
$\F_0R(A)\!\;\F_0^*=\vartheta(A_+)\otimes 1_{\HS}$ imply that
\begin{align*}
W_--1
=-2\pi i\;\!\F_0^*N\big\{\vartheta(A_+)\otimes1_{\H_{\sigma/2}}\big\}B\;\!\F_0
&=-2\pi i\;\!\F_0^*\big\{\vartheta(A_+)\otimes 1_\HS\big\}NB\;\!\F_0+K\\
&=R(A)\;\!\F_0^*(-2\pi iNB)\;\!\F_0+K,
\end{align*}
with $K\in\K(\H)$. Comparing $-2\pi iNB$ with the usual expression for the scattering
matrix $S(\lambda)$ (see for example \cite[Eq.~(6.2)]{Kur73}), one observes that
$-2\pi iNB=\int_{\R_+}^\oplus\d\lambda\,\big(S(\lambda)-1\big)$. Since $\F_0$ defines
the spectral representation of $H_0$, one deduces that $W_--1=R(A)(S-1)+K$. The
formula for $W_+-1$ follows then from the relation $W_+=W_-\;\!S^*$.
\end{proof}

%--------------------------------------------------------------------------------------
\section*{Acknowledgements}
%--------------------------------------------------------------------------------------

R.T.d.A. thanks the University of de Lyon 1 for its kind hospitality in September and
November 2012. The authors are also grateful for the hospitality provided by the
Institut Mittag-Leffler (Djursholm, Sweden) in December 2012.

%--------------------------------------------------------------------------------------
%\bibliography{../bibliographie/bibliographie}
%--------------------------------------------------------------------------------------

\def\polhk#1{\setbox0=\hbox{#1}{\ooalign{\hidewidth
\lower1.5ex\hbox{`}\hidewidth\crcr\unhbox0}}}
\def\polhk#1{\setbox0=\hbox{#1}{\ooalign{\hidewidth
\lower1.5ex\hbox{`}\hidewidth\crcr\unhbox0}}} \def\cprime{$'$}
\def\cprime{$'$}

%--------------------------------------------------------------------------------------

\end{document}